\documentclass[twocolumn,nofootinbib,floatfix,preprintnumbers,amsmath,amssymb,amsthm,groupedaddress,superscript
address]{revtex4}
\usepackage{mciteplus}
\usepackage{multirow} 
\usepackage{dsfont}
\usepackage{graphicx} 
\usepackage{epstopdf} 
\usepackage{slashed}
\usepackage{hyperref} 
\usepackage[all]{hypcap} 
 
\newcommand{\TeV}{\,\mathrm{TeV}}  
\newcommand{\GeV}{\,\mathrm{GeV}}

\newcommand{\secref}[1]{Section~\ref{#1}}

\newcommand{\lem}{Lemma}
\newtheorem{lemma}{\lem}
\newcommand{\lemref}[1]{\lem~\ref{#1}}
\newcommand{\lemrefs}[1]{\lem s~\ref{#1}}

\newenvironment{proof}[1][Proof]{\begin{trivlist}
\item[\hskip \labelsep {\bfseries #1}]}{\end{trivlist}}

\newcommand{\definmath}[2] {\def#1{\ifmmode#2\else$#2$\fi}}
\definmath{\sps}{{{\mathrm{Sp\bar{p}S}}}}

\definmath{\mttwo}{m_{T2}}
\definmath{\invfb}{{\mathrm{fb}^{-1}}}
\definmath{\invpb}{{\mathrm{pb}^{-1}}}  
\definmath{\pslash}{{\slashed{p}}}
\definmath{\ptslash}{{\slashed{p}_T}}
\definmath{\ptmiss}{{\slashed{\bf p}_T}}
\definmath{\pt}{{p_T}}
\definmath{\Pt}{{{\bf p}_T}}
\definmath{\cht}{\tilde{\chi}}
\definmath{\ntlone}{{{\cht^0_1}}}
\definmath{\gluino}{{{\tilde{g}}}}
\definmath{\ttbar}{{t\bar{t}}}

\def\lsim{\mathrel{\rlap{\lower4pt\hbox{\hskip1pt$\sim$}}
    \raise1pt\hbox{$<$}}}                
\def\gsim{\mathrel{\rlap{\lower4pt\hbox{\hskip1pt$\sim$}}
    \raise1pt\hbox{$>$}}}                

\begin{document}   
\title{The race for supersymmetry: using \mttwo\ for discovery}
\date{\today}

\author{Alan J. Barr}
\email{a.barr@physics.ox.ac.uk}
\affiliation{Denys Wilkinson Building, Keble Road, Oxford, OX1 3RH, United Kingdom}

\author{Claire Gwenlan}
\email{c.gwenlan@physics.ox.ac.uk}  
\affiliation{Denys Wilkinson Building, Keble Road, Oxford, OX1 3RH, United Kingdom}

\begin{abstract} 
We describe how one may employ a very simple event selection, using only the kinematic variable
\mttwo, to search for new particles at the LHC. 
The method is useful when searching for evidence of models (such as $R$-parity conserving supersymmetry) 
which have a $Z_2$ parity and a weakly-interacting lightest parity-odd particle.
We discuss the kinematic properties which make this variable an excellent
discriminant against the great majority of Standard Model backgrounds.
Monte Carlo simulations suggest that this approach could be used to discover supersymmetry with 
somewhat smaller integrated luminosities (or perhaps lower center-of-mass energies) 
than would be required for other comparable analyses.
\end{abstract}
\maketitle 
\section{Introduction}\label{sec:intro}

The search for new particles cannot be divorced from measurement of their masses. 
Indeed new particles have usually been discovered by measuring differential distributions 
in variables sensitive to their mass 
\cite{Arnison:1983rp,Aubert:1974js,*Arnison:1983mk,*Abachi:1995iq,*Abe:1995hr}.
A notable example was the successful $W$ boson search at the CERN \sps\ collider 
\cite{Arnison:1983rp,Arnison:1983zy,*Banner:1983jy} using the transverse mass, $m_T$.
Using the right variables can improve the discrimination between signal and background,
increase the statistical power and thus expedite discovery.

One of the primary goals of the Large Hadron Collider (LHC) is to search for physics beyond the Standard Model. 
Many well-motivated models of new physics \cite{Dimopoulos:1981zb,*Cheng:2002ab,*Cheng:2003ju}  
have a $Z_2$ parity (for example $R$-parity in supersymmetry)
which predicts the presence of stable, presumed weakly-interacting, particles which would contribute to the astronomical dark 
matter. Such particles must be produced in pairs and would be invisible to the LHC detectors.


In this paper we explain why the \mttwo\ variable, which was originally proposed \cite{Lester:1999tx} for
supersymmetric particle mass determination in hadron colliders, is also useful in {\em searches} for such models.\footnote{Some of 
these 
observations were originally made in a talk at CERN by the authors in May 2007.
Since then, the method has been employed, for example in \cite{Aad:2009wy},
but the underlying properties which represented the motivation for the approach have not previously been published.}

We start by recalling the definition of the transverse mass and its generalisation \mttwo,
and outline their defining properties (\secref{sec:defs}).
In \secref{sec:props} we show that experimentally-interesting kinematic configurations 
often lead to near-minimal values of \mttwo.  In \secref{sec:example} we explain why these properties should
cleanly separate new physics signals from the great majority of the Standard Model background events.
A Monte Carlo simulation of a typical di-jet channel can be found in \secref{sec:sim}.
We discuss the results and conclude in \secref{sec:discussion}.


\section{Definitions}\label{sec:defs}

The familiar transverse mass, which was used for the $W$ boson discovery, 
is defined by
\begin{equation}\label{eq:mtdef}
m_T^2 \equiv m_v^2 + m_i^2 + 2 \left( e_v e_i - {\bf v}_T \cdot {\bf q}_T \right) ,
\end{equation}
where $m_v$ and $m_i$ are the masses of the visible and invisible systems respectively
(e.g. the electron and neutrino for the $W$ case),
${\bf v}_T$ and ${\bf q}_T$ are their momenta in the transverse plane,
and the transverse energies are defined by
$$
e_v^2 = {m_v^2 + {\bf v}^2_T} \quad \mathrm{and} \quad e_i^2 = {m_i^2 + {\bf q}^2_T}.
$$
The transverse momentum ${\bf q}_T$ of the unobserved neutrino can be 
inferred from momentum conservation in the directions perpendicular to the beam: 
it is assumed to be equal to the missing transverse momentum, $\ptmiss$ 
which is defined to be the the negative sum of the transverse momenta of all the
particles observed in the detector (i.e. including hadrons in the $W$ case).

Why was $m_T$ useful for the $W$ search?
The defining property of $m_T$ is as follows: for events in which a visible and an invisible system 
originate from the decay of a single parent particle of mass $m_0$, 
and when the correct daughter masses are used, $m_T$ satisfies
\begin{equation}\label{eq:mtprop}
m_T \leq m_0
\end{equation}
by construction\footnote{If the width of the particle and the experimental resolutions are assumed to be small.}, 
with equality when the rapidity of all invisible daughters is equal to that of the visible system. 
The transverse mass gives the lowest event-by-event upper bound on $m_0$. 
The maximum value of $m_T$ over all events therefore delineates the boundary 
between allowed and forbidden values of the parent mass, $m_0$.

For the $W$ boson search the allowed values of $m_T$ could therefore span a range $ m_< \lsim m_T \lsim m_W$ for signal events.
The lower bound is given by $m_<=m_v+m_i$, which is zero if one assumes (as was done implicitly in the original $W$ search papers) 
that both $m_v$ and $m_i$ are zero.
Other processes which could have led to lepton + \ptmiss\ final states, for example  
leptonic decays of $\tau$ leptons or $B$ hadrons were also constrained
by similar inequalities, but these backgrounds have smaller values of $m_0$. There was therefore a region
bounded approximately by $m_B \lsim m_T \lsim m_W$ in which there were very few background events
and a preponderance of the signal.\footnote{In \cite{Arnison:1983rp} the lower end of this distribution was 
removed because of threshold cuts on the electron \pt\ and \pslash.}
\par
A corresponding construction can be made for the types of LHC search described in the introduction.
We are now interested in the case in which {\em two} new particles are produced, 
each of which decays to a set of (one or more) visible and (one or more) invisible daughters.
We label the two branches of the decay with the superscripts $^{(1)}$ and $^{(2)}$ to distinguish them,
and for the purposes of this paper assume that the visible particles can be unambiguously
assigned to one or other parent, though this need not be done in general \cite{Lester:2007fq,*Nojiri:2008hy}. 
According to our assumptions, each invisible system should include one dark matter candidate.

As was the case for single-particle decays, there will be some values of parent and (visible and invisible) daughter masses which 
are consistent with the observed momenta, and others which are not.

The \ptmiss\ constraint needs some modification compared to the single decay case. 
For a pair of invisible systems the transverse momentum of each the two invisible-particle systems is 
not individually known, but the sum is constrained by ${\bf q}_T^{(1)} + {\bf q}_T^{(2)} = \ptmiss$.
One can construct the boundary between the allowed and disallowed regions using 
the variable \mttwo\ \cite{Lester:1999tx}
\begin{equation}\label{eq:mttwodef}
\mttwo(v_1, v_2,\ptmiss, m_i^{(1)}, m_i^{(2)}) \equiv \\ 
     \min_{\sum{\bf q}_T = \ptmiss} \left\{ \max\left( m_T^{(1)}, m_T^{(2)} \right) \right\} .\
\end{equation}
For decays of two parents, each of mass $m_0$, \mttwo\ gives the lowest upper bound on $m_0$ consistent with the observed
momenta and the hypothesised daughter masses.
The maximum of \mttwo\ over events therefore delineates the boundary of the domain (in the space of parent and daughter masses) 
which is consistent with the observed momenta \cite{Cheng:2008hk}.
It is this bounding property which makes it ideally suited for measuring 
supersymmetric particle masses at the LHC 
\cite{Cho:2007qv,*Barr:2007hy,Gripaios:2007is,*Cho:2007dh,Cheng:2008hk,Allanach:2000kt,*Barr:2002ex,*Barr:2003rg,*Weiglein:2004hn,
*Meade:2006dw,*Baumgart:2006pa,*Fuks:2006me,*Cho:2007fg,*Conlon:2007xv,*Cho:2007uq,*Cheng:2007xv,*Fuks:2007us,*Ross:2007rm,*Nojiri
:2007pq,*Nomura:2008pt,*Tovey:2008ui,*:2008gva,*Han:2008gy,*Cho:2008cu,*Serna:2008zk,*Hamaguchi:2008hy,*Kitano:2008sa,*Barr:2008ba
,*Nojiri:2008vq,*Choi:2008pi,*Brandt:2008nq,*Wienemann:2008jk,*Cho:2008tj,*Burns:2008va,*Barr:2008hv,*Read:2008zz,*Konar:2008ei,*H
isano:2008ng,*Bai:2009it,*Burns:2009zi,*Matchev:2009iw,Lester:2007fq,*Nojiri:2008hy}. It has also recently been used to measure 
the top quark mass at the Tevatron \cite{Lee:2009eg,*cdftop}.

We note that the Lorentz 2+1 dimensional vectors $v = (e_v,{\bf v}_T)$ and $i = (e_i,{\bf q}_T)$ can each represent single 
particles or multi-body systems.
For example $v$ could represent an individual lepton, or a jet or a di-jet pair. 
Similarly each of the invisible-particle systems may represent a single massive invisible particle (such as a \ntlone) or a single 
particle of negligible mass (neutrino) or a set of invisible particles (e.g. $\ntlone + \nu$).
For composite systems, the composite 2+1 vectors can be formed
from the vector sum of the 2+1 vectors of their individual 
constituents.\footnote{Alternative combinatorial procedures are available
for the {\em visible} system \cite{gaitordraft}. 
}

\section{Relevant properties}\label{sec:props}

Later in \secref{sec:example} and \ref{sec:sim}, we shall find that many background processes predict
values of $\mttwo$ near the global minimum, $m_<$.
The value of this minimum is therefore of some interest. 
It follows from \eqref{eq:mttwodef} that 
\begin{equation}\label{eq:mmin}
m_< = \max ( m_<^{(1)},m_<^{(2)} ) \ ,
\end{equation}
where $m_<^{(n)} = m_v^{(n)}+m_i^{(n)}$ is the global minimum of each of the two individual $m_T^{(n)}$.

In order to construct \mttwo\ we require hypotheses for the invisible particle(s) masses $m_i^{(n)}$.
In a search we usually don't know {\it a priori} whether any invisible particles are massive
(or if so, what those masses are) so we choose to set the invisible-particle mass hypotheses $m_i$ to zero.
This is the only value guaranteed to preserve \eqref{eq:mtprop} 
if the invisible particles' true masses are unknown.\footnote{If and when signals suggesting new invisible particles are observed, 
their masses can later be determined, for example by constructing variables 
sensitive to the kink in the maximum of \mttwo\ is as a function of the input value of $m_i$ 
\cite{Cho:2007qv,*Barr:2007hy,Gripaios:2007is,*Cho:2007dh}.} It is also the appropriate choice
for background processes where the invisible particles are ($\sim\,$massless) neutrinos.

In the rest of this section we show that the observable
$\mttwo(v_1,v_2, \ptmiss, 0, 0)$ must adopt small values for a variety of kinematic configurations.
In \secref{sec:example} we use these results to show that such configurations are 
satisfied (at least approximately) by very many of the Standard Model processes which represent backgrounds to example searches.

\begin{lemma}\label{lem:equalpair}
When a pair of particles are produced, both with mass $m_0$,
and each parent decays to a visible system $v$ and an invisible system $i$,
then $ \mttwo(v_1,v_2, \ptmiss, 0, 0) \leq m_0 $.
\end{lemma}
\begin{proof}
Consider the chained pair of inequalities 
$\mttwo(v_1,v_2, \ptmiss, 0, 0) \leq \mttwo(v_1,v_2, \ptmiss, m_i^{(1)}, m_i^{(2)}) \leq m_0 $.
The first inequality follows since each $m_T^{(n)}$ is a monotonic function of the (non-negative) parameter $m_i^{(n)}$. 
The second inequality is satisfied by \mttwo\ by construction: one of the partitions of the 
missing momentum is the correct one, and for that partition each $m_T\leq m_0$ by \eqref{eq:mtprop}.
\end{proof}

\begin{lemma}\label{lem:unequalpair}
When two particles are produced with different masses $m_1$ and $m_2$ 
and each parent decays to a visible system $v$ and an invisible system $i$
then $\mttwo(v_1, v_2, \ptmiss, 0, 0) \leq \max(m_1,m_2)$.
\end{lemma}
\begin{proof}
As for \lemref{lem:equalpair} except that at the correct \ptmiss\ partition we can only be sure that 
$m_T^{(1)}\leq m_1$ and $m_T^{(2)} \leq m_2$ and so $\max(m_T^{(1)},m_T^{(2)}) \leq \max(m_1,m_2)$ 
for that partition.
\end{proof}

\begin{lemma}\label{lem:ptzero}
When ${\bf v}_T^{(n)} = {\bf 0}$ and $m_v^{(n)} = m_i^{(\overline{n})}=0$ then $\mttwo = m_< $; for $n\in\{1,2\}$.
\end{lemma}
\begin{proof}
Without loss in generality let $n=2$ (and so $\overline{n}=1$). Now $m_T^{(2)} = m_<^{(2)}\ \forall\ {\bf q}_T^{(2)}$.
There exists a partition of \ptmiss\ with ${\bf q}_T^{(1)} = {\bf 0}$ 
for which $m_T^{(1)} = m_<^{(1)}$.
The result follows from \eqref{eq:mttwodef} and \eqref{eq:mmin}.
\end{proof}

An important example of \lemref{lem:ptzero} is the decay of a single parent of mass $m_0$.
The second (non-existent) visible system is represented by $v_2=(0,{\bf 0})$.
With $m_i^{(1)}=0$, \lemref{lem:ptzero} shows that $\mttwo=m_< = \max(m_i^{(2)},m_v^{(1)})$ 
where the second equality is a result of \eqref{eq:mmin}.

\begin{lemma}\label{lem:ptmisszero}
When $\ptmiss = {\bf 0}$ and $m_i^{(1,2)} = 0$ then $\mttwo = m_<$.
\end{lemma}
\begin{proof}
For $\ptmiss = {\bf 0}$ there exists a trivial partition of the missing momentum with 
${\bf q}_T^{(1)} = {\bf q}_T^{(2)} = {\bf 0}$.
For that partition, $m_T^{(1)} = m_<^{(1)}$ and $ m_T^{(2)} = m_<^{(2)}$; the result follows from \eqref{eq:mttwodef} and 
\eqref{eq:mmin}.
\end{proof}

\begin{lemma}\label{lem:ptmissparallel}
When $m_i^{(1,2)}=0$, $m_v^{(n)}=0$ and $\ptmiss \parallel {\bf v}_T^{(n)}$ then $\mttwo = m_<$ ; for $n\in\{1,2\}$.
\end{lemma}
\begin{proof}
Without loss of generality let $n=1$.
There exists a partition of \ptmiss\ with ${\bf q}_T^{(1)}=\ptmiss$ and ${\bf q}_T^{(2)}={\bf 0}$.
Each $m_T$ is equal to its global minimum by \eqref{eq:mtdef}, and the result follows from \eqref{eq:mmin}.
\end{proof}

\begin{lemma}\label{lem:ptmisssum}
When $m_i^{(1,2)}=m_v^{(1,2)}=0$, and $\ptmiss$ can be expressed as the sum $\ptmiss=\sum_k x_k {\bf q}_T^{k}$ 
for some real non-negative $x_k$ then $\mttwo=m_<=0$.
\end{lemma}
\begin{proof}
For the partition of \ptmiss\ given by these $x_k$, ${\bf q}_T^{(n)} \parallel {\bf v}_T^{(n)}$ simultaneously 
for both $n\in\{1,2\}$. For that partition each $m_T^{(n)}=m_<^{(n)}=0$.
The result follows from \eqref{eq:mttwodef} and \eqref{eq:mmin}.
\end{proof}

\begin{lemma}\label{lem:constituent}
If either or both visible systems are composite, then
each of the results in \lemrefs{lem:equalpair} -- \ref{lem:ptmisssum} 
also hold when either or both the composite
Lorentz 2+1 vectors $v$ are replaced by any subset of their (respective) constituents.
\end{lemma}
\begin{proof}
\lemrefs{lem:ptzero} -- \ref{lem:ptmisssum} follow exactly as before.
For \lemrefs{lem:equalpair} and \ref{lem:unequalpair} it is sufficient to show that
\mttwo\ cannot be increased when a constituent is removed from a composite (visible) system.

Now $m_T^2$ is the inner product of a sum of Lorentz 2+1 vectors
$$m_T^2 = \left(v+i\right)^2 = \left(\sum c_k\right)^2,$$ 
where the sum runs over all the constitutent $c_k$, whether visible or invisible.
Separating one of the visible particles, $c_j$ and letting $\sigma=\sum_{k\ne j} c_k$,
$$m_T^2 = \sigma^2 + v_j^2 + 2 \sigma\cdot v_j\ .$$
Since each of the constituent vectors is time-like (or null)
each of the terms is non-negative and $\sigma^2\leq m_T^2$. 
But $\sqrt{\sigma^2}$ is precisely the transverse mass one obtains if $v_j$ is omitted,
so the transverse mass cannot be increased by omitting a particle.
This result also holds at the partition chosen by the minimisation of 
\eqref{eq:mttwodef} so \mttwo\ is never increased by omitting one 
(or by induction more than one) of the visible particles.
\end{proof}

Physically relevant configurations will rarely (if ever) conform to the precise
configurations described in \lemrefs{lem:equalpair} -- \ref{lem:constituent}; parent particles will be off-mass-shell, 
Standard Model particles will have small but non-zero masses, \ptslash\ will never be exactly zero
nor vectors exactly parallel.
However, since $m_T$ is a continuous function of its inputs,
kinematic configurations close to those described above will have upper bounds close to these idealised cases.

\section{Example}\label{sec:example}

The properties described in \secref{sec:props} turn out to be particularly useful for various LHC searches.
We demonstrate the reasons why by examining the concrete example of pair-production of heavy, 
strongly-interacting particles of mass $m_0$, each of which decays to a 
light-quark jet and an invisible particle. The characteristic final state is thus two (usually high-\pt) jets 
with significant \ptslash.
Examples of models which could lead to such a final state can be found in \autoref{tab:models}.

\begin{table}[ht]
\begin{center}
\begin{tabular}{l l c l}
Supersymmetry : \hspace{5mm} & $\tilde{q}\,\bar{\tilde{q}} $&$\to$&$ q \ntlone\, \bar{q} \ntlone$ \\
UED           : \hspace{5mm} & $q_1\,\bar{q}_1             $&$\to$&$ q \gamma_1\, \bar{q} \gamma_1 $ \\
Leptoquarks   : \hspace{5mm} & $LQ\, \overline{LQ}         $&$\to$&$ q \nu\, \bar{q}\bar{\nu}$ .
\end{tabular}
\caption{\label{tab:models}
Examples of models of new physics producing di-jets in association with \ptslash.}
\end{center}
\end{table}

\begin{table*}[tbh]
\begin{minipage}{0.75\linewidth}
\begin{tabular}{| l | l | l |}
\hline
Process & $\mttwo(v_1, v_2, \ptmiss, 0, 0)$  & Comments \\
\hline
QCD di-jet $\to$ hadrons          & $=\max m_j$ by \lemrefs{lem:equalpair},\ref{lem:ptmisszero} & \\
QCD multi jets $\to$ hadrons     & $=\max m_j$ by \lemref{lem:ptmisszero}  & \\
\ttbar\ production               & $=\max m_j$ by \lemref{lem:ptmisszero}  & fully hadronic decays\\
                                 & $\le m_t$ by \lemrefs{lem:equalpair},\ref{lem:constituent}   & any leptonic decays\\
Single top / $tW$                & $=\max m_j$ by \lemref{lem:ptmisszero}  & fully hadronic decays\\
                                 & $\le m_t$ by \lemrefs{lem:unequalpair},\ref{lem:constituent} & any leptonic decays\\
Multi jets: ``fake'' \ptslash    & $=\max m_j$ by \lemref{lem:ptmissparallel} & single mismeasured jet\footnotemark[1]\\
                                 & $=\max m_j$ by \lemref{lem:ptmisssum} & two mismeasured jets\footnotemark[1]\\
Multi jets: ``real'' \ptslash    & $=\max m_j$ by \lemref{lem:ptmissparallel} & single jet with leptonic $b$ 
decay\footnotemark[1]\\
                                 & $=\max m_j$ by \lemref{lem:ptmisssum} & two jets with leptonic $b$ decays\footnotemark[1]\\
 $Z \to \nu\bar{\nu}$            & $=0$ by \lemref{lem:ptzero}      & \\
 $Z\,j \to \nu\bar{\nu}\,j$      & $=m_j$ by \lemref{lem:ptzero}      & one ISR jet\footnotemark[1] \\
 $W \to \ell\nu$ \footnotemark[2]& $=m_\ell$ by \lemref{lem:ptzero}      & \\
 $W\,j \to \ell\nu\,j$ \footnotemark[2] & $\le m_W$ by \lemref{lem:unequalpair}      & one ISR jet\footnotemark[1]\\
 $WW \to \ell\nu \ell\nu$ \footnotemark[2] & $\le m_W$ by \lemref{lem:equalpair}   & \\
 $ZZ\to\nu\bar{\nu}\nu\bar{\nu}$ & $= 0$   by \lemref{lem:ptzero}      & also $= m_j$ for one ISR jet\footnotemark[1]\\
\hline
$LQ\, \overline{LQ}\, \to q\nu \bar{q}\bar{\nu} $             & $\le m_{LQ} $       & \multirow{3}{*}{{\huge\}}~i.e. can take 
large values}\\
$\tilde{q}\,\bar{\tilde{q}} \to q \ntlone\, \bar{q} \ntlone $ & $\le m_{\tilde{q}}$ & \\
$q_1, \bar{q}_1 \to q \gamma_1, \bar{q} \gamma_1 $ & $\le m_{q_1}$ & \\
\hline
\end{tabular}
\footnotetext[1]{~Assuming that the relevant jet(s)
are identified with one of the two visible particle systems $v^{(n)}$.}
\footnotetext[2]{~Even if the lepton is mistaken for a jet.}
\caption{\label{tab:backgrounds}
Examples of Standard Model backgrounds, various signals, and associated values of \mttwo.
In the cases marked $^{b}$, there are more restrictive upper bounds
if the lepton is assumed not to be misidentified as a jet.
}
\end{minipage}
\end{table*}

When considering how to search in such channels, a variety of levels of sophistication can be envisioned.

In a typical cut-based search one would select events with large \ptmiss\ and two high-\pt\ jets;
require sufficiently large azimuthal angles $\Delta\phi_j$ between
each jet and the \ptmiss\ vector to reduce backgrounds from ``fake''
(from the mismeasurement of one of the jet energies) or ``real'' (from neutrinos in jets) \pslash\ from the jet;
and apply further cuts to reduce background processes
(such as leptonic \ttbar\ decays) which produce neutrinos and high-\pt\ jets with larger $\Delta\phi_j$ \cite{Aad:2009wy}.
The remaining events should be signal-enriched, and so, 
provided that the residual background contributions can be understood,
one can check for an excess of events not explained by the Standard Model.
The most difficult task is to understand the contribution of the background to the selected events --
for which a various techniques have been proposed \cite{Aad:2009wy,cmsphystdr}.
Applying such a set of many single-variable cuts is fairly straightforward 
but is somewhat wasteful of signal events
(which one can assume will be in scarce supply at the time of any interesting discovery).

At the other extreme, the formally optimal search method (needing the fewest events)
would require calculation of the likelihood for all events for every signal and background hypothesis.
This is often prohibitively difficult in practice even if all such hypotheses can be identified and modelled.

A pragmatic approach, intermediate in complexity,
is to find a single, easily-calculable observable which, while not optimal in the formal 
statistical sense, still gives very good discrimination between the majority of signal-like and background-like events, based on 
some rather general principles such as relativistic kinematics.
In what follows we shall find that \mttwo\ is just such an observable.

For the example search channels, we identify the two visible systems $v^{(n)}$ with the 
two highest-\pt\ jets.
Our jets, though massive, should have a sufficiently small mass that $m_j \approx 0$ is a good approximation.
For the reasons discussed in \secref{sec:props} we set both $m_i^{(n)}$ to zero.
For these choices, $m_<$ is small (equal to the larger of the $m_j^{(n)}$) and the mass conditions 
required for \lemrefs{lem:ptzero} -- \ref{lem:constituent} are 
satisfied in the $m_j\to 0$ limit.

\autoref{tab:backgrounds} lists a variety of Standard Model processes which
form backgrounds to the search channel.
From \lemrefs{lem:equalpair} -- \ref{lem:constituent} 
we expect to find restrictive upper bounds on \mttwo\ for many of these.
Processes 
with small \ptmiss;
with small jet \pt; 
with small $\Delta\phi_j$;
or from production of one or two Standard Model particles;
all bound $\mttwo$ from above.
Such processes constitute the majority of both the ``physics'' and ``detector'' backgrounds 
to this channel so one can remove the vast majority of the background simply by requiring large \mttwo.
No additional cuts are required -- other than perhaps modest trigger requirements on jet \pt\ or on \ptmiss.

An example ``physics'' background is \ttbar\ production.
This is pair-production of equal-mass particles, so \mttwo\ is bounded above by the mass of the top quark $m_t$ 
by \lemref{lem:equalpair}. One could argue that this is not a very strict bound,
but it is much smaller than masses of the particles typically being searched for, and the top is the heaviest known Standard Model 
particle, so other similar background processes, e.g. hadronic $\tau^+\tau^-$ decays should adopt smaller values of \mttwo\ again.

An example of a detector-induced background is ``fake'' \ptslash\ --
events with no invisible particles, but where a substantial fraction of 
the energy of one of the leading jets is lost. This can happen when energy is
deposited in inactive material the detector (e.g. cables, services or support structures).
Such pathological events cannot necessarily be cleanly identified.
In these cases the detector usually measures a jet with Lorentz 2+1 vector $j^\prime = \alpha j$ ($0<\alpha<1$),
and one gains a contribution to \ptmiss\ of $(\alpha-1){\bf j}_T$.
In the absence of any other source of missing momentum, 
$\ptmiss \parallel {\bf j}_T$, so $\mttwo\to m_<$ by \lemref{lem:ptmissparallel} or \ref{lem:ptmisssum}.
Similar arguments apply to heavy-quark jets where leptonic decays
lead to production of neutrinos close to the jet axis.

The other backgrounds in \autoref{tab:backgrounds} are also forced to small values of 
\mttwo. The least restrictive is $\leq m_t$ since the top is (assumed here to be) the heaviest Standard Model particle.

What values of \mttwo\ are expected for any new particles?
At the upper end it is clear from \lemref{lem:equalpair}, 
that $\mttwo \leq m_0$ for the processes in \autoref{tab:backgrounds}.
One needs a significant number of events with 
$\mttwo>m_t$ to have a signal region which is relatively free of background.
Now if the correct value of $m_i$ were to be used then the upper bound ($m_0$) would be saturated
since there is a significant density of states with 
$\mttwo(v_1, v_2, \ptmiss, m_i, m_i)$ close to $m_0$.
We have chosen to input the lowest conceivable value $m_i\to 0$
rather than the true invisible-particle mass, so the argument does not 
prove that the bound is saturated. 
However one can see from \eqref{eq:mtdef} that, provided $m_i \ll |{\bf q}_T|$, then events which are close to maximal
when the true $m_i$ is used will also remain close when we replace this by $m_i=0$.

We therefore expect to find a large number of signal events, 
and very little background, in the region approximately bounded by 
$m_t \lsim \mttwo \lsim m_0$, where $m_0$ is the mass of the new particle.


\section{Simulation}\label{sec:sim}
\begin{figure}[Htp]
 \begin{center}
  \includegraphics[width=\linewidth,height=0.8\linewidth,trim= 12mm 2mm 8mm 2mm,clip]{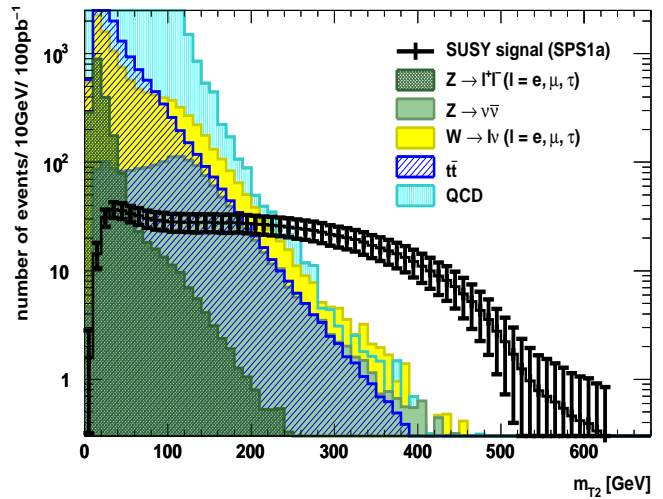}
\caption{
Distribution of \mttwo\ for events with two or more jets with $\pt > 50\GeV$ (and no other cuts).
For the signal point the squark masses are in the range $500 \lsim m_{\tilde{q}} \lsim 600\GeV$ and the gluino mass is close to $600\GeV$.
}
\label{fig:simulation}
\end{center}
\end{figure}

\begin{table*}[bth]
\begin{minipage}{0.75\linewidth}
\begin{tabular}{| l | c | c |}
\hline
{sample} & LO cross section (pb) & number of events \\ \hline
QCD ($17<p_T<35\GeV$)     & $9.44\times 10^8$ & $6\times 10^7$\\
QCD ($35<p_T<70\GeV$)     & $5.99\times 10^7$ & $5\times 10^7$\\
QCD ($70<p_T<140\GeV$)    & $3.45\times 10^6$ & $5\times 10^7$\\
QCD ($140<p_T<280\GeV$)   & $1.57\times 10^5$ & $5\times 10^7$\\
QCD ($280<p_T<560\GeV$)   & 5280 & $4\times 10^7$\\
QCD ($560<p_T<1120\GeV$)  & 116 & $3\times 10^7$ \\
QCD ($1120<p_T<2240\GeV$) & 1.11 & $2\times 10^7$ \\
QCD ($p_T>2240\GeV$)      & $1.15\times 10^{-3}$ & $5\times 10^6$\\
$t\bar{t}$  & 231 & $2\times 10^7$ \\
$W \rightarrow \ell\nu$+jets & 17,100 & $1.5\times 10^8$ \\
$Z \rightarrow \ell^\pm\ell^\mp$+jets  & 1780  & $9\times 10^7$ \\
$Z \rightarrow \nu\nu$+jets    & 3440 & $5\times 10^7$ \\ \hline
SUSY signal (SPS1a)  & 13.5     & $2\times 10^6$ \\ 
\hline
\end{tabular}
\caption{\label{tab:samples}
Signal and background processes together with their leading order cross-sections, and number of events generated for each.
}
\end{minipage}
\end{table*}


To illustrate the results of \secref{sec:props} and \secref{sec:example} we generate Monte Carlo signal and background samples 
with Herwig++~2.3.2 \cite{Bahr:2008rn,*Bahr:2008mn}. 
The background processes simulated are QCD, $t\bar{t}$, $W\rightarrow \ell\nu$+jets, $Z\rightarrow \ell^+\ell^-$+jets, 
and $Z\rightarrow \nu\bar{\nu}$+jets. 
The contribution from diboson+jets is expected to be very small \cite{Aad:2009wy} so is not considered here.  
In order to generate sufficient QCD events in the high-\pt\ region, eight samples were generated in slices of the $p_T$ of the 
hard scatter.    
For the SUSY signal, the SPS1a point \cite{allanach:2002sp} is used ($m_0=100\GeV, ~m_{1/2}=250\GeV, ~A_0=-100\GeV, ~\tan\beta=10, 
~\mu > 0$), as calculated by SPheno~2.2.3 \cite{Porod:2003sp}.  
Table~\ref{tab:samples} lists the leading order cross sections calculated by Herwig++, and the number of events generated for each 
of the processes considered.

We cluster hadrons (and $\pi^0$s) with fiducial pseudorapidity ($|\eta|<5$) and momentum ($\pt>0.5\GeV$) into jets 
using the fastjet~\cite{Cacciari:2005hq} implementation of the anti-$k_T$ algorithm \cite{Cacciari:2008gp}. 
We use the $E$ combination scheme and set $R=0.4$ and $p_T^{\min}=10\GeV$.
To simulate the detector effects we smear the majority ($1-\epsilon$) of the jet energies by a Gaussian probability density 
function of width 
$$\sigma(E)/E_j = \left(0.6/\sqrt{E_j[\GeV]}\right)\, \oplus\, 0.03$$ 
where $E_j$ is the unsmeared jet energy. This resolution is typical of one of the general-purpose LHC detectors 
\cite{Aad:2009wy,cmsphystdr}.
Since the tails of the \ptslash\ distribution are often dominated by badly mismeasured jets,
we simulate pathological energy-loss by applying a different smearing function to the remaining fraction $(\epsilon=0.1\%)$ of the 
jets\footnote{
\cite{Aad:2009wy} suggests a larger value of $\epsilon\sim 1\%$. We find a smaller value
better matches the tails of the \ptmiss\ and \mttwo\ distributions found in full simulation.
The detailed form of the transfer function clearly needs to be determined from the collision data, 
but our findings are not materially altered by changing epsilon from 0.1\% to 1\%. 
} with probability density:
$$
P(E) = \left\{ \begin{array}{ll} 2E/E_j^2\ & \mathrm{for} ~ (0<E<E_j) \\ 0 & \mathrm{elsewhere}\end{array}\right. \ .
$$
The missing transverse momentum is calculated from the negative vector sum of the visible fiducial hadrons (including $\pi^0$)
and is corrected for the jet smearing. 
We impose the simple requirement that each event contains at least two jets with $\pt>50\GeV$. 
We then take the two highest \pt\ jets as $j_{1,2}$ and calculate $\mttwo(j_1, j_2, \ptmiss, 0, 0)$,
for all events. 
We normalise to 100 \invpb\ (using the leading order cross sections for both signal and background).
The resulting distribution can be seen in \autoref{fig:simulation}.

As predicted, the region $m_t\lsim \mttwo \lsim m_{\tilde{q}}$ is dominated by signal events.
The great majority of the Standard Model background events are found at small \mttwo, 
as required by the arguments of \secref{sec:example}.
The remaining backgrounds for which \mttwo\ is not bounded above by the arguments in \secref{sec:example} are:
$Z\to\nu\bar{\nu}$ in association with two hard jets from initial state radiation (ISR);
$W\to\mathrm{leptons}$ plus two hard ISR jets; and three- (or more-) jet production where one of the jets looses a very large 
amount of energy, and that jet is not one of the two highest \pt\ jets input to \mttwo.
For the signal point examined, and with the level of sophistication used for our simulations,
we find that the residual backgrounds for $\mttwo\gsim m_t$ are well below the supersymmetric signal predictions.


\section{Discussion}\label{sec:discussion}

The numerical simulations of \secref{sec:sim} confirm the analytic results of \secref{sec:props};
\mttwo\ adopts small values -- equal to or less than the mass of Standard Model particles
-- for the vast majority of events from the background processes.
The region $\mttwo \gsim m_t$ is signal-dominated; for SPS1a a single cut requiring
$\mttwo > 230 \GeV$ results in a statistical significance $S/\sqrt{S+B}$ of 15,
for 100~\invpb\ of integrated luminosity at $\sqrt{S}=10\TeV$.
This is a higher significance than is found when applying 
to the same simulated events the selections of comparable multi-cut-based analyses 
e.g. \cite{Aad:2009wy,Randall:2008rw,cms:susy08005}. 

Applying a single-variable kinematic selection based solely on \mttwo\ would be sufficient to separate the signal from the 
background. 
In practice \mttwo\ is unlikely to be calculated by the trigger algorithms (certainly not at the first level) 
so some thresholds for the jet \pt s
and the \ptslash\ will be required in order to keep trigger rates within the available bandwidth,
particularly at higher instantaneous luminosities. 
Our simulations show that adding typical trigger requirements has the effect of removing events at small \mttwo\
(as would be expected from \lemrefs{lem:ptzero},\ref{lem:ptmisszero}). Reasonable thresholds leave the 
large \mttwo\ region -- and hence the statistical significance -- practically unchanged.
The lower \mttwo\ end of the distribution can be recovered using lower-threshold prescaled triggers,
so it can be used as a control region.

The remaining backgrounds at large \mttwo\ are dominated by states originating from more than two 
(massive or high-\pt) Standard Model particles (e.g. $Zjj$, $Wjj$, $t\bar{t}j$, $jjj$). 
The bounds of \secref{sec:props} do not apply to these 
since the variable \mttwo\ was developed for the explicit two-parent case.
To reduce the residual contribution from these $n>2$-particle topologies 
one might consider using, rather than \mttwo, an $n$-parent generalisation of the transverse mass:
\begin{equation}\label{eq:mtndef}
m_{Tn} \equiv \min_{\sum {\bf q}_T = \ptmiss} \{\max m_T^{(i)}\}\quad \mathrm{with}\quad i\in\{1,\ldots,n\}.
\end{equation}
Using \eqref{eq:mtndef} $n$-body background topologies (and also $m<n$-body topologies by a generalisation of 
\lemref{lem:unequalpair}) 
would be bounded from above by the mass of the heaviest parent.
The problem would be that one would no longer expect pair-produced signal topologies 
(such as those in \autoref{tab:models}) to obtain $m_{Tn}$ values close to the mass of the new heavy parent.
Indeed a generalisation of \lemref{lem:ptzero} shows that for the simplest $Z_2$ signal case, 
di-jet + \ptslash, $m_{T3}$ is forced towards the small value $(\max m_j)$ so one loses the discrimination power of \mttwo.
Such $n>2$ generalisations are therefore only likely to be appropriate when $n$ heavy signal particles are
expected to be produced.

Our introduction focused on the simplest decay topologies of pair-produced heavy particles such as
$\tilde{q}\,\bar{\tilde{q}} \to q \ntlone\, \bar{q} \ntlone$
but our simulations also show excellent discrimination in more complicated cases using the same \mttwo\ variable.
Decay sequences with many steps, or indeed individual multi-daughter decays, such as
$\gluino\gluino\to q \bar{q} \ntlone q \bar{q} \ntlone$ (via three-body decays),
could have been input to \mttwo\ in a number of ways. 
There is no unique choice for constructing the two `visible particle systems' from the $(k>2)$ decay products.
We could have chosen to form two composite systems (two di-jet systems for the gluino example), and use these as the visible 
inputs to \mttwo\ \cite{Cho:2007qv,*Barr:2007hy,Lester:2007fq,*Nojiri:2008hy}. 
This construction would have provided a large number of signal events close to the kinematic boundary
and so would be appropriate for mass determination.
However forming di-jet composite objects produces visible systems which no longer have $m_v\approx 0$, 
so one would lose the desirable properties of \lemrefs{lem:ptzero}, \ref{lem:ptmissparallel} -- \ref{lem:constituent}.
For our search, even though cascade or multi-body decays are expected, 
we have still choosen to form $\mttwo(j_1, j_2, \ptmiss, 0, 0)$ using {\em only} the two highest \pt\ jets --- precisely as 
for the simpler di-jet topology.
This way we retain the desirable properties (\lemrefs{lem:ptzero}, \ref{lem:ptmissparallel} -- \ref{lem:constituent}) for the 
backgrounds. 
The \mttwo\ distribution for the signal can still extend up to large values (close to $m_\gluino$ for the three-body decay),
albeit with fewer near-maximal events than would be found by using composite two-jet visible systems.
By inputting only the highest two \pt\ jets to \mttwo, not only do we take advantage of the background rejection properties, but 
we also combine many signal channels together, forming a larger sample of signal events (which therefore enjoys a larger 
statistical significance).


Of course even having found a good discriminating variable one is not 
absolved from the need to understand the residual background contributions.
Previous studies \cite{cmsphystdr,Aad:2009wy} show that the rates of a wide variety of Standard Model processes can be measured in 
control regions using the LHC collision data. 
These are then extrapolated to the signal region using techniques which themselves are validated by Monte Carlo simulation.
Our \mttwo\ signal region has a non-negligible contribution from vector boson
production in association with two (or more) jets from initial state radiation.
The dependence of such ISR jets on parameters such as $\sqrt{s}$ and $\eta$ is therefore of considerable interest; jets at large 
$|\eta|$ could, for example, provide an ISR-dominated control region.\footnote{It has recently been suggested that ISR jets could 
themselves provide a good indicator of mass scale \cite{Papaefstathiou:2009hp}.}
The fact that there is a large number of background events in the low $\mttwo$ region $(\lsim m_t)$
might actually turn out to be rather helpful -- one can measure the total background contributions 
from this control region with high precision. 

Complementary measurements will clearly be needed to disentangle the various background processes,
so many observables will be used by the LHC experiments. 
Our analysis and simulations suggest that \mttwo, the natural kinematic observable for pair-produced particles, ought to be at the 
front of the queue.

\acknowledgements

We are grateful to C.G. Lester and to B. Gripaios for encouraging the wider 
dissemination of these results, and to Mihoko Nojiri and the IPMU, Tokyo, Japan, 
for hosting the workshop which prompted those discussions. 
Further thanks to C.G. Lester, B. Gripaios and to A. Pinder, and K. Matchev for helpful comments and discussions.
AJB and CG are supported by Advanced Fellowships from the UK Science and Technology Facilities Council.

\bibliography{mt2fordiscovery}
\end{document}